\documentclass[letterpaper,12pt]{article}

\usepackage[margin=1in]{geometry}

\usepackage{amsthm,amsmath,amsfonts}
\usepackage[normalem]{ulem}
\usepackage{enumerate}
\usepackage[title,toc,titletoc]{appendix}
\usepackage[boxruled]{algorithm2e}
\usepackage{graphicx}
\usepackage{setspace}
\usepackage{subcaption}
\usepackage{multirow}
\usepackage{floatrow}
\usepackage{hyperref}

\newfloatcommand{capbtabbox}{table}[][\FBwidth]

\usepackage{natbib}
\usepackage{blindtext}

\DontPrintSemicolon

\usepackage{graphicx}
\usepackage{float}
\graphicspath{{./figures/}}

\usepackage[dvipsnames]{xcolor}

\usepackage{url}

\newtheorem*{lemma*}{Lemma}

\providecommand{\keywords}[1]
{
  \small	
  \textbf{\textit{Keywords---}} #1
}

\title{Bayesian averaging of computer models with domain discrepancies: a nuclear physics perspective}
\author{Vojtech Kejzlar,$^{1}$ L{\'e}o Neufcourt,$^{1,2}$  Taps Maiti,$^{1}$ and Frederi Viens$^{1}$ \\
\small $^{1}$Department of Statistics and Probability, Michigan State University \\
\small $^{2}$Facility for Rare Isotope Beams, Michigan State University \\}
\date{}
\begin{document}
\maketitle

\begin{abstract}
This article studies Bayesian model averaging (BMA) in the context of competing expensive computer models in a typical nuclear physics setup. 
While it is well known that BMA accounts for the additional uncertainty of the model itself, we show that it also decreases the posterior variance of the prediction errors via an explicit decomposition. We extend BMA to the situation where the competing models are defined on non-identical study regions. Any model's local forecasting difficulty is offset by predictions obtained from the average model, thus extending individual models to the full domain. We illustrate our methodology via pedagogical simulations and applications to forecasting nuclear observables, which exhibit convincing improvements in both the BMA prediction error and empirical coverage probabilities. 
\end{abstract}


\section{Introduction}
\label{Sec:Intro}

Interest for model averaging arises in situations with several competing models available to solve the same or similar problems, 
and  no single model can be selected at a desired level of certainty. 
This is a common scenario in many scientific fields concerned with modeling complex systems.
One of the historically dominant solutions is Bayesian Model Averaging (BMA) \citep{KassRaftery1995,BMA, GIBBONS2008973}, which is a natural Bayesian framework when facing a finite or countable number of alternative models. The seminal review work by \cite{Geweke} introduced BMA in econometrics and later in other fields such as political and social sciences; 
BMA has also been applied to the medical sciences \citep{med-Vi14, med-SCH16}, 
ecology and evolution \citep{ecol-Si14,ecol-Ho15}, 
genetics \citep{gen-Vis11, gen-Wen15}, astrophysics \citep{Parkinson2013}, 
fluid dynamics \citep{Radaideh}
and machine learning \citep{ML-Mer11, ML-Hernandez2018}.

Contemporary developments in computing capabilities have meanwhile brought new modeling perspectives, and the rapid surge of models implemented on a computer, which we shall refer to as \textit{computer models}, opened several challenges to BMA. 
Nuclear structure provides stimulating examples \citep{NuclearLandscape,ElectricDipole,NeutronDripline,Neufcourt18-2,Olsen18} which illustrate the canonical difficulties one faces with BMA of computer models that might require hundreds to thousands of core hours for a single evaluation.

Consider a general situation where experimental measurements $(x_i,y_i)_{i=1}^n$ of a physical process $x\mapsto y(x)$ are used to predict its values $y^* = y(x^*)$ at a new input value $x^*$. 
In the nuclear physics context one can typically think of $x=(Z,N)$ marking the isotopes defined by the proton number $Z$ and neutron number $N$, with the output $y(x)$ being the corresponding value of an observable such as a nuclear mass or radius, and $y(x^*)$ extrapolation of this observable.
The Bayesian inference approach is to consider any quantity of interest $\Delta$ under its posterior probability given the data $y = (y_1, \dots, y_n)$ and model $\mathcal{M}$. 
This is simply obtained from Bayes' formula as long as one can express the likelihood $p(y|\Delta, \mathcal{M})$ and provide a prior distribution $\pi(\Delta| \mathcal{M})$. 
Establishing the likelihood $p(y|\Delta, \mathcal{M})$ is done through a model $\mathcal{M}$. 
The quantity $\Delta:=y^*=y(x^*)$ is typically of universal interest, but $\Delta$ can more generally represent any latent quantity at any level of the model.
Various well-known methods such as the Metropolis-Hastings algorithm \citep{Metropolis} or more advanced Monte Carlo Methods such as Hamiltonian Monte-Carlo, or  No U-Turn Sampler \citep{NUTS} can be applied to obtain samples from the posterior distribution $p(\Delta|y, \mathcal{M})$.  

Now, let us consider a scenario with $K$ competing models $\mathcal{M}_1, \dots, \mathcal{M}_K$, and let us assume that there exists a \emph{true} model ${M}$.
The Bayes formula can be applied  to express the posterior probability of a model $\mathcal{M}_k$, given observations $y$, as
\begin{equation} \label{eq-posteriors-model}
p(\mathcal{M}_k|y)
= \frac{p(y|\mathcal{M}_k)\pi(\mathcal{M}_k)}{\sum_{\ell=1}^K p(y|\mathcal{M}_\ell) \pi(\mathcal{M}_\ell)}
\propto_k p(y|\mathcal{M}_k)\pi(\mathcal{M}_k),
\end{equation}
where the symbol $\propto_k$ indicates proportionality with respect to $y$ for fixed $k$. 
Here, the difficulty lies in evaluating the \emph{evidence integral} $p(y|\mathcal{M}_k)$ for each model, and $\pi(\mathcal{M}_k)$ represents the the prior probability that $\mathcal{M}_k$ is the true model $M$; the BMA posterior distribution for any quantity of interest $\Delta$ can then be derived as
\begin{equation} \label{eq-posterior2}
p(\Delta|y) 
= \sum_{k=1}^K p(\Delta|y,\mathcal{M}_k) p(\mathcal{M}_k|y).
\end{equation}
This formula expresses that the actual posterior probability of an observable $\Delta$ is the average of $\Delta$'s posterior distributions given each model, weighted by the model posterior probabilities.
In other words, \eqref{eq-posterior2} is simply a mixture of $K$ distributions, which makes sampling from the BMA posterior density immediate once we obtain posterior samples under each model. 
In particular, the posterior mean of $\Delta$ is given by (see \cite{Park2011})
\begin{equation} \label{BMA-mean}
\mathbb{E}[\Delta|y] = \sum_{k=1}^K  \mathbb{E}[\Delta|y, \mathcal{M}_k] p(\mathcal{M}_k|y),
\end{equation}
and the well-known conditional variance formula yields the posterior variance of $\Delta$, given $y$, as
\begin{eqnarray} \label{BMA-variance}
\mathbb{V}ar[\Delta|y] 
& = & \sum_{k=1}^K p(\mathcal{M}_k|y) \mathbb{V}ar[\Delta|y, \mathcal{M}_k] + \mathbb{V}ar[\mathbb{E}(\Delta|y,M)|y].
\end{eqnarray}
Note that the  term 
$\mathbb{V}ar[\mathbb{E}(\Delta|y,M)|y]$ 
 is the variance of a function of the discrete random variable $M$, which accounts for the model uncertainty. This model uncertainty is not accounted for by individual models. Its inclusion thus allows for a more honest uncertainty quantification (UQ).

In this work, we first describe the Bayesian methodology for analyzing individual computer models including calibration and evidence integral computation with a generalization to account for experimental data coming from several sources (several observables). Then, we perform a systematic analysis of the prediction errors, establishing that BMA is the optimal linear combination (projection) in the $L^2$ sense under the posterior probability distribution, among all possible mixtures of models. Motivated by recurrent scenarios in nuclear structure, we subsequently extend BMA to situations when different models constrain different subsets of the data. 
Lastly, we present a set of pedagogical examples as well as a real-data nuclear-physics application of the BMA methodology highlighting its benefits in terms of improvement of the prediction accuracy and UQ. We frequently use nuclear physics terminology, but we hope the paper will be transparent to scientists from other quantitative disciplines.

\section{Review of Bayesian analysis of computer models}
\label{Sec:BACM}

\subsection{Calibration}

Over the past two decades, there has been a considerable amount of research dedicated to Bayesian calibration of computer models starting with the seminal work of \cite{KoH} with extensions and applications provided by \cite{Higdon04, Higdon08, Higdon15}. We briefly discuss here the standard Bayesian framework and refer to \cite{KoH} for full treatment.

Let us consider a single computer model $f(x,\theta)$ relying on a parameter vector $\theta$ of dimension $d$ and an input variable $x$ in a finite dimensional space (for simplicity).
Model calibration corresponds to determining the unknown and hypothetical \emph{true} value $\theta^*$ of the parameter $\theta$,
at which the physical process $\zeta(x)$ would satisfy $\zeta(x) = f(x,\theta^*) + \delta(x)$; $\delta(x)$ is the systematic discrepancy of the model whose form is generally unknown. The term "calibration" is broader than the term "estimation", which can imply the use of some well-established statistical methodology. Herein, our notion of calibration is a specific Bayesian estimation methodology, which include a full evaluation of uncertainty for every parameter. Thus we can write the complete statistical model as
\begin{equation} \label{eq-model}
y_i = f(x_i,\theta) + \delta(x_i) + \sigma(x_i) \epsilon_i.
\end{equation}
at each of the $n$ observations $y_1,\dots, y_n$. For brevity, we use $y_i$ for $y(x_i)$. The term $\sigma(x) > 0$ is a scale function to be parametrized and inferred,
and $\epsilon_i$ are independent standard random variables
representing measurement errors, which we assume to be Gaussian.

 One may expect the function $f(x,\theta)$ representing the output of the computer model to be deterministic, given a model and its parameters. 
 This holds in the case of \emph{inexpensive} computer models, where the evaluation of $f(x,\theta)$ takes a reasonably "constant" time.
 For computationally \emph{expensive} models 
 the evaluations of $f(x,\theta)$ cannot be reasonably performed 
 at calibration runtime, and need to be done beforehand, typically on a grid.
 Hence it has become a common practice to emulate the computer model by a Gaussian process $\mathcal{GP}(m(x, \theta), k((x,\theta),(x',\theta')))$ with mean function $m$ and covariance function $k$. 
 In this setup the data also include a number $N$ of runs of the computer model at pre-determined points $\{(\widetilde{x}_1, \widetilde{\theta}_1), \dots, (\widetilde{x}_N, \widetilde{\theta}_N)\}$. Finer approaches typically decompose $m$ on a dense family of basis functions (wavelets, Fourier, polynomials) across the domain of $x$ and $\theta$. As for the covariance function, the most popular choices are the quadratic exponential kernel, the Matern kernels, and fractional Gaussian noise \citep{RasmussenWilliams}. 

The discrepancy function $\delta(x)$ represents the systematic error between the computer model and the physical process. 
While it is intrinsically deterministic, a Gaussian process model is typically imposed for inference. 
However, its inference can be tricky because the coupling of the discrepancy term and calibration parameters makes $\theta^*$ non-identifiable in general. 
Indeed, $\delta(x) = \zeta(x) - f(x, \theta)$ yields the same distribution for $y(x)$ for any choice of $\theta$. Several authors have pointed this out 
and proposed various methods to mitigate the problem including \cite{Ha,Plumee,Wu1,Wu2,Vlid}. Our main goal here, nonetheless, is not to correct the identification of $\theta^*$, but a prediction of new observations with honest uncertainties quantification. For prediction, the identifiability of calibration parameter is not a particular concern.

It can also appear natural to carry out the calibration based on a linearization of the computer model, i.e. in an setup where $f(x,\theta)=x^T\theta$. 
This simple framework is valid in practice, up to a reparametrization, as an approximation of a function $f(x,\theta)$, as long as one can \textit{a priori} localize the parameters around a specified value $\theta_0$, such as the maximum likelihood estimator or the solution of a classical least-squares fit of the computer model. Indeed, in the range of a first order Taylor expansion of $f(x,\theta)$ around $\theta_0$, one can replace the computer model by a linear function of $\theta$ 
\begin{equation}
f(x,\theta) \approx f(x,\theta_0)+\nabla_\theta f(x,\theta_0)^T\cdot (\theta-\theta_0),
\label{eq:approx}
\end{equation}
as long as one has access to the gradients $\nabla_\theta f$.
This method has the advantages to be intuitive, computationally inexpensive, and can be used either as a surrogate in the procedure outlined herein for calibration of inexpensive computer models, or to provide a mean function of a GP emulator for computationally expensive model. We, therefore, shall call this approximation \textit{rapid calibration via linearization}.

\paragraph{Several observables for the same model.} 
The simple notational form of model \eqref{eq-model} obscures the situation where $y$ mixes data of different nature, which can have critical impact on both calibration and model averaging. This can represent different 
types of observables in nuclear physics. Many nuclear structure models are for example optimized over a set of binding energies which represent the minimum energy to disassemble nucleus into unbound neutrons and protons, and a set of root mean square (rms) proton charge radii which can be seen as a measure of the size of an atomic nucleus. 
Suppose that the same model is fitted to $q$ types of observables $y^{o_1},\dots, y^{o_q}$. To highlight the dependency on observable types, we can write
\begin{equation}
y_i^{o_l} = f^{o_l}(x_i, \theta) + \delta(x_i) + \sigma_l(x_i) \epsilon_i, \hspace{1cm} l=1,\cdots, q,
\end{equation}
and the corresponding likelihood can be computed naturally for each data $y_i$ corresponding to an observable $y^{o_l}$. Allowing different types of observables doesn't raise significant differences with the standard case,
but a particular attention is to be brought to the noise scaling parameter $\sigma_l(x)$ which can now vary across observables. In this case the weighting between different observables is done through the relative importance of the scaling functions $\sigma_l(x)$ which can either be 
left as parameters or fixed (e.g., to reported experimental errors, or estimated theoretical variations). 
In the context of nuclear physics 
we refer to the discussion in \cite[Section 3]{JPG}.

\subsection{Computing the evidence integral} 
\label{sub-Bayes-factor}

Suppose the $k^{th}$ model is parametrized  by a vector
$\phi_k\in\mathbb{R}^{d_k}$ for $k = 1, \dots, K$ 
that consists of both calibration parameters and hyperparameters (i.e. parameters used for defining prior distributions for other parameters) 
with likelihood $p(y|\mathcal{M}_k, \phi_k)$. The \emph{evidence integral} \citep{KassRaftery1995,evint-AA13,Fragoso} of the model $\mathcal{M}_k$ is defined as
\begin{equation}
p(y|\mathcal{M}_k) = \int_{\phi_k} p(y|\phi_k,\mathcal{M}_k)\pi(\phi_k|\mathcal{M}_k)d\phi_k. \label{eq-I}
\end{equation}

The numerical evaluation of evidence integrals is challenging in practice, and requires approximation. A natural idea, most commonly used in the literature, and which we have adopted in our applications:, is to use Monte Carlo approximation
\begin{equation}
\widehat{p(y|\mathcal{M}_k)} = \frac{1}{n_{MC}}\sum_{i}p(y|\phi_k^{(i)}, \mathcal{M}_k).
\end{equation}
where $\phi_k^{(i)}$ are i.i.d. samples from the prior $\pi(\phi_k| \mathcal{M}_k)$ for  $i=1,\ldots,n_{MC}$. While this Monte Carlo computation yields reasonable results, it requires separate evaluations of the likelihood on new samples from the prior $\pi(\phi_k| \mathcal{M}_k)$which can be very costly in computing time.

Another frequently used method is the Laplace approximation, which relies on the fact that the integration \eqref{eq-I} has a closed form in the case of a linear regression with Gaussian noise.
It corresponds to a second order Taylor approximation of the log-likelihood around its maximum, which makes the likelihood becomes Gaussian. The Laplace method typically gives very good results for very peaked likelihoods.
We refer the reader to \cite{KassRaftery1995} for an exhaustive survey of classical methods used to compute evidence integral.

\section{BMA and prediction error}

BMA is only one of various natural ways to deal with several alternative models and to account for model uncertainty, but it does have the property of reducing the Posterior Mean Square Error (PMSE) of prediction of a new observation $y^*$. In this section, we illustrate this property in a clear and concise way.

Let us, for simplicity of notation, consider two competing models $\mathcal{M}_1$ and $\mathcal{M}_2$ - the treatment of multiple models follows from a similar argument, and our verbal descriptions below in this section occasionally refer to the general case without further comment. Denote $\widehat{y_1^*}:=\mathbb{E}[y^*|y, \mathcal{M}_1]$ and $\widehat{y_2^*}:=\mathbb{E}[y^*|y, \mathcal{M}_2]$ as the posterior means of $y^*$ under each model, and let $\widehat{y^*}:=\mathbb{E}[y^*|y]$. We also define $p_k:=p(\mathcal{M}_k|y)$ for $k =1,2$ for the posterior probability of each model. Thus the BMA posterior mean estimator (\ref{BMA-mean}) can be written as $\widehat{y^*} = p_1 \widehat{y_1^*} + p_2 \widehat{y_2^*}$.
The PMSE of $y^*$ is then defined naturally as $\mathbb{E}[(\widehat{y^*} - y^*)^2|y]$ and has the following decomposition.
\begin{lemma*} \label{lemma-factor} For every $\lambda_1, \lambda_2\geq0$ satisfying $\lambda_1 + \lambda_2 = 1$, we have
\begin{equation}
\small
\mathbb{E}[(y^*- \widehat{y^*})^2 |y] 
= 
\mathbb{E}[(y^*- \lambda_1 \widehat{y_1^*} - \lambda_2 \widehat{y_2^*})^2 |y]
- 
[ (\lambda_1-p_1) \widehat{y_1^*} - (\lambda_2-p_2) \widehat{y_2^*} ]^2
\end{equation}
\end{lemma*}

This lemma, proved in Appendix A, shows explicitly that the PMSE of the BMA predictor is smaller than the PMSE associated with any convex combination $\lambda_1 \widehat{y_1^*} + \lambda_2 \widehat{y_2^*}$ of the each of the two models' posterior means. It also measures how much smaller it is, and 
shows that equality holds 
as soon as the convex coefficients $\lambda_k$ are equal to the posterior probabilities $p_k$ of each model, $k=1,2$. Now, applying the Lemma twice, with $(\lambda_1, \lambda_2) = (1, 0)$ and with $(\lambda_1, \lambda_2)=(0, 1)$ we obtain the following dual expressions for the PMSE or the BMA predictor, involving each individual model's PMSE, showing how much smaller the former is compared to the two latter:    
\begin{equation}
\small
\mathbb{E}[(y^*-  \widehat{y_1^*})^2 |y] - p_2^2 (\widehat{y_1^*} - \widehat{y_2^*})^2
= \mathbb{E}[(y^*- \widehat{y^*})^2 |y] 
= \mathbb{E}[(y^*-  \widehat{y_2^*})^2 |y] - p_1^2 (\widehat{y_1^*} - \widehat{y_2^*})^2. 
\label{eq-replace}
\end{equation}

To be more descriptive about these properties, first we record the optimality of BMA's PMSE as the following inequality: 
\begin{equation}
\mathbb{E}[(y^*- \widehat{y^*})^2 |y] \le \mathbb{E}[(y^*- \widehat{y_i^*})^2 |y], \quad i = 1,2.
\end{equation}
The previous inequality is the clearest way to state that the BMA estimator (\ref{BMA-mean}) produces the smallest prediction error, in the PMSE sense, among all the individual models' posterior mean estimators, as long as all those models are used in creating the BMA estimator. We interpret this as a translation of the fact that each model that goes into creating the BMA estimator necessarily ignores model uncertainty. Note that this says nothing about how the BMA estimator would compare to a model not used in its definition.

But more can be said. As we mentioned implicitly, the previous inequality is a weaker statement than the statement in the Lemma, since the latter covers optimality over all convex combinations of the original models, not just the individual models themselves: the Lemma shows additionally that BMA achieves the following minimum
\begin{equation}
\big(p(\mathcal{M}_k|y)\big)_{k=1,2} := {\arg \min}_{\lambda\in[0,1]^2 : \lambda_1 + \lambda_2 = 1} \mathbb{E}[\big(y^*-(\lambda_1 \widehat{y_1^*} + \lambda_2 \hat{y_2^*})\big)^2 | y].
\end{equation}
Hence, the BMA estimator is actually optimal over all convex combinations 
of the individual estimators $\widehat{y_1^*}$ and $\widehat{y_2^*}$.

We can also express the reduction of the PMSE for the BMA estimator, compared to the best (lowest) PMSE among all of the individual models',  as 
\begin{equation} \label{R-bma-m}
r_{BMA}^2 := 1 - \frac{\mathbb{E}[ (\widehat{y^*} - y^* )^2 | y ]}{\min_i\mathbb{E}[ (\widehat{y_i^*} - y^* )^2 | y ]}.
\end{equation}\noindent
In the specific case of two competing models, if we assume for instance that the 'best' model is $\mathcal{M}_2$,
we can obtain an even more explicit expression for $r^2_{BMA}$ which provides the relative gain attained by BMA, namely
\begin{equation} 
r_{BMA}^2 = p(\mathcal{M}_1|y)^2 \frac{(\widehat{y_1^*} - \widehat{y_2^*})^2}{\mathbb{E}[ (\widehat{y_2^*} - y^* )^2 | y]}.
\end{equation}
Below in the Application section, we denote the sample version of the expression in \eqref{R-bma-m} as $\hat{r}_{BMA}^2$, which we will use to evaluate the performance of BMA quantitatively.

To finish this section, we decompose the quantity $\mathbb{E}[ (\widehat{y^*} - y^* )^2 | y]$ 
against the residuals  
$(\widehat{y_i^*} - y^* )$, $i = 1,2$ from each individual model assuming $p_1, p_2 > 0$. This is easily done by symmetrizing formula \eqref{eq-replace} via reintroducing $y^*$ to identify these residuals, and then taking another conditional expectation with respect to $y$ to avoid an expression which depends on unobserved data, which doesn't affect the left hand side. We obtain 
\begin{equation} \label{mse-interp}
\begin{aligned}
\mathbb{E}[(y^*- \widehat{y^*})^2 |y] & = (p_1 - p_1^2)\mathbb{E}[(y^*-  \widehat{y_1^*})^2 |y] + (p_2 - p_2^2)\mathbb{E}[(y^*-  \widehat{y_2^*})^2 |y]\\
& - (p_1^2 + p_2^2)\mathbb{E}\big[(\widehat{y_1^*}-  y^*)(y^*-  \widehat{y_2^*})|y\big].
\end{aligned}
\end{equation}
Formula \eqref{mse-interp} shows that the PMSE of the BMA estimator is an explicit linear combination of the prediction errors of estimators for each constituent model, those that do not consider model uncertainty, but that one must subtract a coupling correction term on the right hand side of \eqref{mse-interp}. 

It is interesting to note that the weights in the aforementioned linear combination can be interpreted as the variances of Bernoulli random variables with the posterior model probabilities $p_1$ and $p_2$ as their success probabilities. Also note that, since these variances $p_k - p_k^2 < p_k$, the  linear combination is not convex, but is smaller. The correction term is not necessarily a subtraction of a positive term, but it is likely to be so in some generic cases, for example, when both individual models have significant biases in opposite directions for prediction of $y^*$. This is particularly interesting when both models have similar posterior performances. Indeed, then, both values of $p_k$ will be close to $1/2$, which minimizes the values of $p_k - p_k^2 $ for both $k=1,2$. This is a scenario where using the BMA model will significantly improve prediction errors even when each model is competitive compared to the other, regardless of how large the individual models' biases are, and without knowing in what direction they go, as long as the two models are known or assumed to have significant defects that work in opposite directions.  

A sanity check reveals an interesting characteristic of BMA: suppose that $p_1=1$, so that the BMA estimate is given by $\widehat{y^*} = \widehat{y_1^*}$. According to \eqref{mse-interp} we \emph{must} have 
$\mathbb{E}[(y^*- \widehat{y^*})^2 |y]=0$, and further $\mathbb{E}[(y^*- \widehat{y^*})^2]=0$, 
i.e. $y^*= \widehat{y^*} = \widehat{y_1^*}$ a.s. given $y$, 
in other words Model 1 must provide a perfect description of the reality.

The PMSE corresponds to an MSE under the posterior probability. It is of different nature that the MSE one would typically compute when evaluating a model against experimental measurements. Nevertheless it is a good proxy for an actual MSE when posterior predictions are close to actual measurements.

\section{Domain-corrected model averaging} \label{sec-domain-correction}

\paragraph{Motivation.} 
It is easy to imagine scenarios where alternative models
are defined on different subsets of the same input space; 
this can typically arise with local models 
or with numerical models with different constraints.
It is also a usual situation in nuclear physics, for instance for nuclear mass models; as a fact in point, ab initio (also known as $A$-body) models aggregating individual two- and three-body interactions range over lighter nuclei due to contemporary computational limitations, 
while Energy Density Functionals (EDF), 
which relies on mean field approximations,
can cover the whole nuclear chart.
This even occurs for different EDFs \citep{Klupfel08,Kortelainen10}.
This also happens when one considers mixing models produced by the calibration of a given EDF
independently on observables of different types,
if some parameters are not constrained 
by all observables and a particular model has a wider definition domain for a specific observable - 
typically some nuclear models are mostly fitted on masses, and others on masses and radii.
Surprisingly, we have not found in the literature a principled approach to 
adapt BMA to this situation, or how to compare 
models with similar, overlapping, but significantly non-identical domains. 
To address this "domain discrepancy", we present a method which relaxes the requirement that all models cover the same domain, and we name our procedure \emph{domain-corrected BMA}.
Other applications of our framework could include time series with missing data, or different time scales, e.g. in a financial setting where additionally different classes of assets can be treated as observables.

Let us start by considering two models $\mathcal{M}_A$ and $\mathcal{M}_B$, 
which we will also denote by $(A)$ and $(B)$ or merely $A$ and $B$ for simplicity, and assume that they are respectively defined only on different strict subsets $x^{(A)}$ and $x^{(B)}$ of the data.
We denote $y^{(A)}$ and $y^{(B)}$ the corresponding $y$ data as well as $y^{(-A)}$ and $y^{(-B)}$ their respective complements in $y$.
The actual \emph{Bayesian evidence} for each of these models are the probabilities $p(y|A)$ and $p(y|B)$, but these quantities are not clearly defined.
On the other hand $p(y^{(A)}|A)$ and $p(y^{(B)}|B)$, 
 where each model refers only to its original range of validity,
are the evidences of the models corresponding to 
the classical BMA theory which can be computed as detailed in Section \ref{sub-Bayes-factor}. We have:
 \begin{equation} \label{eq-bma-div}
	p(y|A) 
	 =  p(y^{(A)}, y^{(-A)}|A) 
	 =  p(y^{(A)}|A) p(y^{(-A)}|y^{(A)}, A) .
\end{equation}
This expression means that to obtain model $(A)$'s actual Bayesian evidence,  $p(y^{(A)}|A)$ must be multiplied by a corrective factor $p(y^{(-A)}|y^{(A)}, A)$ 
which represents the information one has on $y^{(-A)}$ assuming that model $(A)$ holds and that it does not provide any prediction at the data points in $x^{(-A)}$. 
Note that the distribution $p(y|A)$ is meaningful only to the extend that $y$ -- and thus $y^{(-A)}$ -- is measurable in the underlying probability space, which implies the existence of  underlying distributions $p(y^{(-A)})$ and subsequently of $p(y^{(-A)}|y^{(A)})$ and $p(y^{(-A)}|y^{(A)}, A)$.
To that extent, the problem of averaging models with different domains can be ill posed, if these distributions cannot be defined convincingly.

If the data $y^{(A)}$ and $y^{(-A)}$ are independent, conditionally to model $(A)$, in other words if no information can be gleaned about $y^{(-A)}$ from $y^{(A)}$ or from $(A)$,
i.e. $y^{(A)}$ is unconstrained by $(A)$ and by  $y^{(-A)}$, then it is legitimate to ignore the aforementioned correction factor which should be $p(y^{(-A)}|y^{(A)}, A) = 1$. 
In particular, this is the case if, 
given model $(A)$,  $y^{(-A)}$ is considered deterministically equal to its sample value.
Conversely, setting the corrective factor to
$p(y^{(-A)}|y^{(A)}, A) = 1$ 
outside of this scope is an approximation to such extend, 
and not in general a fair evaluation of the information contained in the "globality" of a model. We shall pragmatically call this approach \emph{BMA with independent model domains}.
Although it has been adopted 
as a natural matter of convenience,
it raises serious safeguards for which 
we cannot find better words than Trotta's ascertainment:

\begin{quotation}
On the other hand, it is important to notice that the Bayesian evidence does not penalize models with parameters that are unconstrained by the data. It is easy to see that unmeasured parameters (i.e. parameters whose posterior is equal to the prior) do not contribute to the evidence integral, and hence model comparison does not act against them, awaiting better data \citep{Trotta08}. 
\end{quotation}

Let us point out as a caveat the other extreme situation that occurs when model $(A)$ does not predict the values $y^{(-A)}$ to exist,
e.g. the mass of a nucleus $X$ which the said model predicts not to exist hence has no physical meaning.
In this case, the model $(A)$ is actually strongly constrained by $y^{(-A)}$,
to the point that $p(y^{(-A)}|A) = 0$, 
yielding $p(y^{(-A)}|y^{(A)}, A) = 0$,
which rules the model $(A)$ impossible 
as long as $y^{(-A)}$ is not empty.

Another tempting option is to restrict 
the domain of interest to the domain common to all models, and simply consider
$p(y^{(A)\cap (B)}|A)$ and $p(y^{(A)\cap (B)}|B)$, 
which can also be obtained from Section \ref{sub-Bayes-factor}. As we ignore even more data, this approach is arguably worse than
setting $p(y^{(-A)}|y^{(A)}, A) = 1$.

Let us illustrate how the assumption of independent model domains approach, 
namely setting $p(y^{(-A)}|y^{(A)}, A) = 1$,
can fail to provide a satisfactory ranking of models in two examples where a model takes a shortcut by 'refusing' to predict challenging points.

\paragraph{Scenario 1.}
Consider the situation where one model $\mathcal{M}_0$ is empty, in the sense that 
$p(y^{(0)} | \mathcal{M}_0) = p(\emptyset|\mathcal{M}_0) = 1$ 
so that $p(\mathcal{M}_0| y) = \pi(\mathcal{M}_0)$. 
On the other hand, any other model which constrains any part of the data will have an evidence lower than $1$ which implies that the model will end up with lower posterior weights when starting from equal prior weights. Thus any predictive model will be deemed inferior to a non-predictive one. 

\paragraph{Scenario 2.} Take two deterministic models $A$ and $B$ with input space (domain of $x$) $\{a,b\}$; 
assume model $A$ has deviation $0$ at location $a$ and $10^{99}$ at location $b$, and that model $B$ has 
deviation $1.001$ at location $a$, but does not predict anything at location $b$. One can easily adjust the numbers to reach an extreme situation (e.g. making $A$'s prediction at location $b$ to be extremely poor) where model $B$ ends up with a much higher Bayes factor than model $A$, while the common sense idea, by which no prediction is a form of extremely poor prediction, would always imply that model $A$ is better than model $B$.\\

These examples show how important it is to acknowledge that a model's inability to make predictions in some locations is not a neutral property. 
The classical BMA approach offers no trade-off: a model withholding its predictions at the most difficult points will always improve its weight.
We now introduce our ``domain-corrected BMA''
where we amend the model weights to account more fairly
for the (in-)ability of a model to provide predictions at locations of interest.

\paragraph{Domain correction with two models.}

Starting from \eqref{eq-bma-div}, 
instead of setting $p(y^{(-A)}|y^{(A)}, A) = 1$ 
which removes the effect of a model's domain in its posterior weights,
we propose the weaker assumption that $p(y^{(-A)}|y^{(A)}, A)$ is independent from the model, i.e. we assume
$$p(y^{(-A)}|y^{(A)}, A) = p(y^{(-A)}|y^{(A)}).$$
This is quite natural if we consider that model $(A)$ implies a distribution $p(y^{(A)}|A)$ 
but provides no information on  $y^{(-A)}$,
leaving $y^{(-A)}$ unconstrained by $(A)$ (see the introduction of this section).
The evidence $p(y|A)$ is now given by
\begin{equation}
p(y|A) \propto_A p(y^{(A)}|A) p(y^{(-A)}|y^{(A)})\pi(A). \label{eq-master-cor-prop-a}
\end{equation}
Our assumption $y^{(A)}\cup y^{(B)}=y$ implies that $y^{(-A)}$ can only be informed by $(B)$. 
Hence 
\begin{equation}
 p(y^{(-A)}|y^{(A)}) 
= p(y^{(-A)}|y^{(A)}, B) 
= p(y^{(-A)}|y^{(A)\cap(B)}	, B)
\end{equation}
which can be written as an explicit integral with respect to model $(B)$'s parameter $\phi_B$,
\begin{equation}
    \int_{\phi_B} p(y^{(-A)}|B, \phi_B)p(\phi_B|y^{(A)\cap(B)}, B) d\phi_B,
\end{equation}
and computed similarly to a classical evidence integral (see Section \ref{sub-Bayes-factor}).

\paragraph{Domain correction in the general case.}

In the general case, each model $\mathcal{M}_k$ constrains a subset $y^{(k)}$ of the data $y$ (for $k=1,\ldots, K$); 
as in the case of two models, $y^{(-k)}$ denotes the complement subset of $y^{(k)}$ in $y$.
We also introduce $y^{(\oslash)}:=\bigcap_{k}y^{(k)}$ as the set of data 
common to all individual models.
Moreover we assume that $y=\bigcup_{k}y^{(k)}$, i.e. every datapoint is covered by at least one model. 
We also assume, up to taking equivalence classes on models
(see Appendix \ref{annex-mixing-domains} for details),
that for each pair of models there exists a chain of models joining them 
where each model $k$ shares a data point in its domain $y^{(k)}$ with each of its neighbours.
Relying on the same principles, we set 
$$p(y^{(-k)} | y^{(k)}, \mathcal{M}_k) = p(y^{(-k)} | y^{(k)}),$$ 
which leads to the model weights $w_k$ of the form
\begin{equation}
p(\mathcal{M}_k | y) 
\propto_k 
p(y^{(-k)} | y^{(k)})
p(\mathcal{M}_k | y^{(k)}). 
\end{equation}

Compared to the two-model case,
the computation of the corrective factors 
shows in general the additional difficulty that, 
when there is more than one model constraining $y^{(-k)}$, the factor
$p(y^{(-k)}|y^{(k)})$ is no longer equal to a single $p(y^{(-k)}|y^{(k)}, \mathcal{M}_k)$ but rather to
the Bayesian average of all models constraining $y^{(-k)}$. 
Hence our domain-corrected BMA corresponds to the
natural, intermediate solution where one replaces the factors of the likelihood corresponding to the missing model predictions by a geometric average of the likelihoods over the models which do produce predictions, based on the predictive models' posterior weights.
We have found that similar ideas have been developed in the broader framework of evidence theory \cite[Section 2.2]{Park2012}.

The notation for a given corrective factor can become quite cumbersome, or could be ambiguous, when model domains have very general intersections, but these corrective factors can still be computed recursively rather than directly. We relegate the calculations of the general case to the appendix, for the sake of readability.

\paragraph{Discussion.}

While it is more liberal than classical BMA, our procedure should still be considered as conservative, in the sense that it is neutral towards (does not penalize or favor) a model not producing any prediction; indeed, 
while a model 
which does not include an existing data point in its domain would have weight 0 in the naive BMA,
it would still have the same weight as the same model which would predict the average value over the other models
when applying the domain-corrected method. 

Our setting calls for a localization of the Bayesian model averaging and selection procedures.
Let us take now model weights $w_k$ as given, either by our domain corrected BMA or any another method. The  posterior average distribution for a prediction $y^*:=y(x^*)$ at a new point $x^*$ can be calculated from
\begin{equation} \label{eq-bma-diff-dom}
    p(y^*|y) 
    \propto \sum_{k=1}^K w_k p(y^*|y^{(k)},\mathcal{M}_k)
    =
    \sum_{x^*\in \mathcal{D}_k}
    w_k p(y^*|y^{(k)},\mathcal{M}_k)
    + 
    \sum_{x^*\notin \mathcal{D}_k} w_k
    p(y^*|y^{(k)},\mathcal{M}_k).
\end{equation}
In the traditional BMA framework, $x^*$ would naturally belong to the domain $\mathcal{D}_k$ of every model $\mathcal{M}_k$, but this does not hold in our framework.
In the model selection problem, picking the best model with highest weight and using it for predictions at all locations would amount to a likelihood of 0 at all locations outside of the domain of the selected model, which is rather extreme.
This calls for a local variation of model selection. For instance, instead of selecting a single best model, one can rank the models such that $\mathcal{M}_0$ is the \emph{best} model and $\mathcal{M}_K$ is the \emph{worst} model and use the following prediction procedure: at each location $x$, pick as true model the model $\mathcal{M}$ with the smallest $k$ which contains $x$ in its domain. One could easily think of more complex procedures. This also suggest local variations of model mixing.

Additionally, this emphasizes that our procedure questions that the models are mutually exclusive (i.e. that there exist a true model $M$ such that each $\mathcal{M}_k$ is of the form $1_{\{M=k\}}$). 
Specifically, our procedure is incompatible with global exclusivity of models in the classical BMA framework, in the sense that the classical BMA assumes that one of the model is true on the whole domain. In our new procedure, we replaced the model prediction with a ``default'' value
at locations where it would otherwise predicts nothing, given by the Bayesian average prediction. 
Consequently our model averages  predict something in any given location, even when the set of predictive models depends on the location. 
Applying the same principle to the germane problem of model selection would lead to producing 
an "optimal" model replaced by the BMA outside its own domain. 

To be absolutely clear, it is of course always possible to bypass the domain issues by restricting estimations to the data which is common to the domains of all models. This circumvents having to deal with models withholding their predictions, but it can leave significant amounts of data unused. That disadvantage can be particularly stark in nuclear physics, when comparing EDF models with $A$-body models, as we mentioned at the start of this section, since these two model classes have such narrow overlap.

\section{Examples and applications}

To illustrate the methodology described in the previous sections, we present several examples in which BMA of computer models leads to reduction in prediction error and improved uncertainty quantification. Our first illustration is a simple yet sensible scenario of averaging two different models of proton potentials. 
The second example is an application of the methodology to nuclear mass models and  nuclear mass data. 
Lastly, we provide a pedagogical application of model averaging to a synthetic dataset which highlights the interest of the domain-corrected BMA.
Each of the examples in this section looks at a situation with several competing models without any prior knowledge of which is better; thus we set the prior model weights to be uniform over the model space. All the posterior samples were computed using a Hamiltonian Markov-Chain Monte-Carlo algorithm. The evidence integrals were approximated using the Monte-Carlo integration.

\subsection{Averaging proton potentials}

In this first example we demonstrate the potential of BMA to improve both prediction accuracy and honesty of uncertainty quantification in a favorable situation 
where we average two models 
associated with different proton potentials.

We consider two single-proton potentials describing 
the average interaction acting on a proton
within the spatial range of a nucleus;
namely, a Wood-Saxon (WS) potential $V_{1}$
representing respectively strong nuclear forces between nucleons (protons and neutrons),
and a Coulomb potential $V_{2}$
representing electromagnetic interactions between protons. 
For a given nucleus, 
which we will take with proton and neutron numbers $Z=100$ and $N=150$ and mass number $A = 250$, 
they can be expressed as 
\begin{align}
V_{1}(r) &= - V_{WS} \frac{1}{1 + e^{\frac{r - R_A}{a}}}, \\
V_2(r) &= - {V_C} \frac{Z}{r}.
\end{align}
Here, $V_{WS} = 50$, $V_{C} = 0.5$ and $a = 0.5$ are fixed parameters, and $R_A =  A^{1/3} \times 1.25$ fm is the radius of the nucleus of interest. 
These two models for energy potentials have the interesting property that both are non-decreasing and vanishing at infinity, while with different speeds,
and can correspond to two phenomenons with different length scales. 
As a matter of fact, strong interactions described by the WS potential are confined to atomic nuclei (several fm $=10 ^{-15}$ m), i.e. they are short-ranged; in contrast electrostatic are long-ranged, i.e. they act on much larger length scales ($>10 ^{-10}$ m) and compete with strong interaction in superheavy elements, causing the so-called Coulomb frustration (see \cite{CoulombFrustration}).
This fact is reproduced in our example where we also expect that $V_{1}$ should be well constrained 
by a dataset of stable nuclei, 
while $V_2$ should reproduce better short-lived superheavy nuclei.
More generally, we have in mind a scenario
where two models have been developed for different subsets of an input domain and are in competition on some common intermediate domain. Both of these modeling approaches are equally confident that they prevail on the intermediate domain, while the truth is somewhere in between. This situation is very realistic despite its simplicity, and we can reasonably expect model mixing to have positive outcomes.

We simulate experimental data
$\{(r_i, y_i)\}_{i = 1} ^ {n}$
at different spatial locations $r_i$, relatively far from the nucleus ($r>R_A$)
following a mixture of the two models. Namely 
\begin{equation} \label{def:potentials_true}
y_i = (1-\omega){V_{1}(r_i)} + \omega {V_2(r_i)} + \epsilon_i
\end{equation}
where $\epsilon_i$ are standard normal errors, and we take $\omega=\frac{1}{2}$. 
Note that in reality observations of the potentials are not available as such, but can be inferred indirectly relatively accurately from experimental nucleonic densities measured in nucleon scattering experiments \citep{ANNI1995}.

In particular, we drew a dataset of $210$ observations generated according to the model \eqref{def:potentials_true} with the locations $r_i$ sampled uniformly over $(R_A, 10)$. 
We further randomly divided the data into a training dataset of 140 observations and kept the remaining 70 observations for testing. 
The two statistical models $\mathcal{M}_1$ and $\mathcal{M}_2$ considered here are given by the respective energy potentials \eqref{def:potentials_true} obtained with $\omega=0$ and $\omega=1$ and additive independent experimental errors distributed according to $N(0, \sigma_j)$ for $j = 1,2$.

\begin{figure}[!h]
\begin{floatrow}\CenterFloatBoxes
\ffigbox{%
  \includegraphics[width=1.02\linewidth]{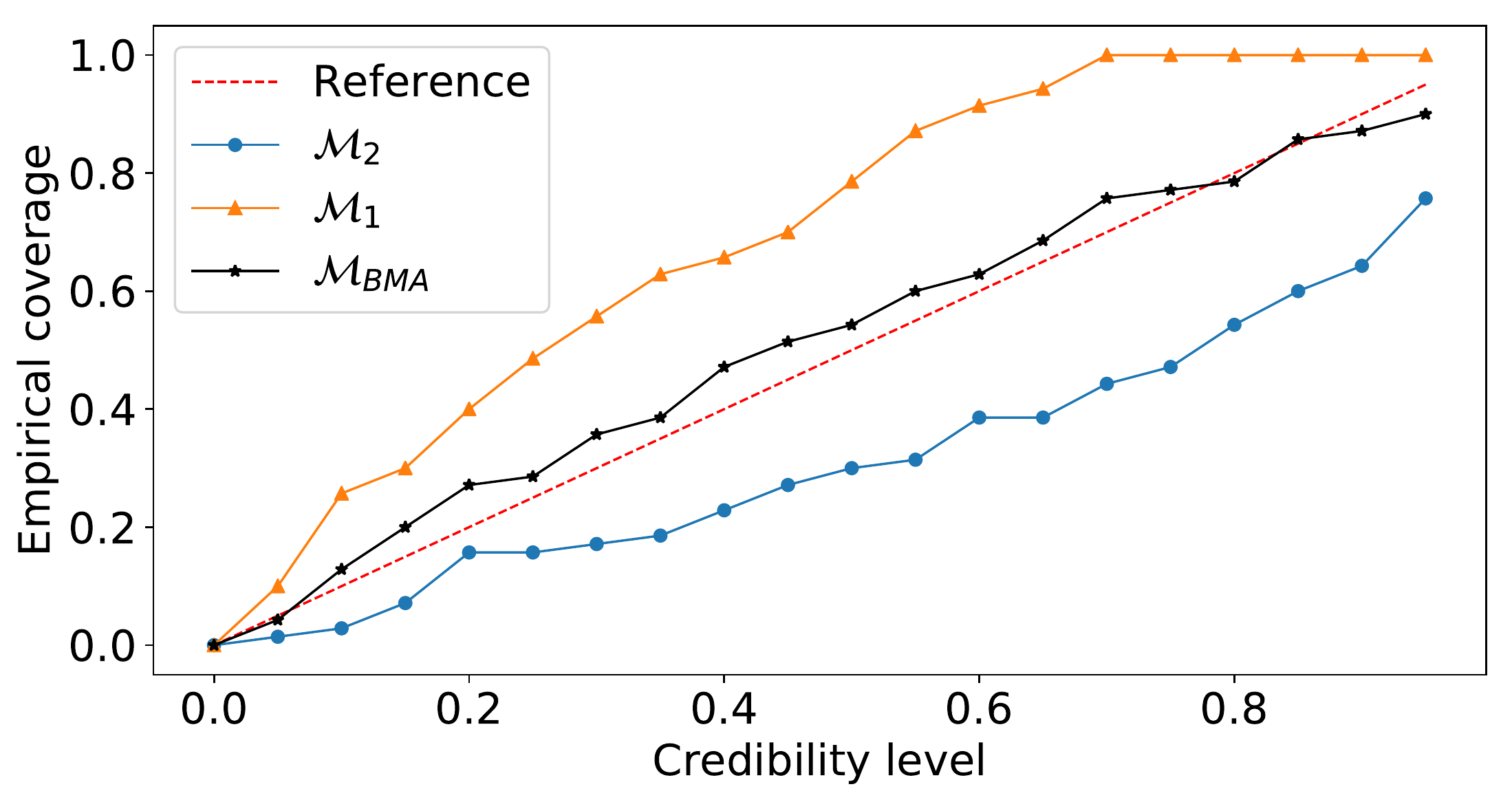}%
}{%
  \caption{Empirical coverage probability for testing dataset ($n=70$, $A = 250$).} \label{fig:ws}%
}
\capbtabbox{%
\setlength{\tabcolsep}{3pt}
\renewcommand{\arraystretch}{0.7}
	\begin{tabular}{lcccc} 
		\hline	\noalign{\smallskip}
		\textbf{Model} & & \boldmath${RMSE}$     & \boldmath$P(\mathcal{M}_k|y)$ & \boldmath$\widehat{r}^2_{BMA}$ \\
		\hline	\noalign{\smallskip}
		$\mathcal{M}_1$   & & $3.540$ & 0.512 &  \boldmath$0.930$    \\
		$\mathcal{M}_2$   & & $3.607$   & 0.488 &   \boldmath$0.933$   \\
		$\mathcal{M}_{BMA}$ &  & $0.935$  &    -    & -   \\
		\hline	\noalign{\smallskip}
	\end{tabular}
}{%
  \caption{RMSE (in MeV) and improvement for the BMA model calculated on the testing dataset ($n=70$, $A = 250$).} \label{table:ws}%
}
\end{floatrow}
\end{figure}

Table \ref{table:ws} shows the estimated residual MSE (RMSE) for the testing dataset.
We can see that this simple example gives significantly better predictions under the BMA posterior mean predictor than for each of the models individually. This is, of course, not a surprise and only shows that the BMA behaves as expected. 
More interesting results can be seen from the angle of the quality of the predictions' UQ. Figure \ref{fig:ws} shows the empirical coverage probabilities (ECP), i.e. the proportion of independent testing points falling into the respective credibility intervals with nominal value given in the horizontal axis - ideally a straight line. 
In contrast with the individual models, the ECP of the BMA posterior predictions matches closely the reference line and provides evidence that accounting for model uncertainty leads to a desired more honest UQ.

\subsection{Averaging nuclear mass emulators in the Ca region} 

An important challenge in nuclear structure is to produce quantified predictions of nuclear observables, such as nuclear masses \citep{McDonnell15}, for all possible pairs $(Z,N)$ of proton number $Z$ and neutron number $N$ which can be bound together in a nucleus. 
Such predictions are of direct interest to guide future nuclear experiments or to feed astrophysical calculations for the abundance of elements in the universe.
The underlying astrophysical processes, such as the rapid neutron capture which produces heavy elements in stellar environments, 
take place far from the region of nuclear stability, where no experimental measurement are available, and these observables have to be extracted from extreme extrapolations of theoretical nuclear models.

In their recent work, \cite{Neufcourt18-2} use Gaussian Processes to model the discrepancies between experimental data and theoretical calculations for several nuclear models based on the density functional theory, and obtain quantified extrapolations for nuclear masses in the Calcium region (at the frontier between experimental and theoretical limits).
They compute a simplified BMA of 9 global mass models \citep{SKMstar,SKP,SLY4,SVMIN,UNEDF0,UNEDF1,UNEDF2} listed in Table \ref{table-rmse-full} defined across the full nuclear landscape from light to superheavy nuclei, thus suitable for extrapolations.  
Their weights, 
formally $w_k\propto p(y^*>0|y,\mathcal{M}_k)$,
are based on each model's probability to assign a positive separation energy $y^*$ to a testing set of nuclei which have been experimentally observed after 2003, 
thus independent from the training set of measured separation energies $y$. 
Here, we compare their results to the full BMA analysis presented in Section \ref{Sec:BACM}
applied to the same framework, this time with weights $w_k\propto p(y|y,\mathcal{M}_k)$. 
Note that all the physical models models are taken here as calibrated and their parameter estimation is not part of our analysis. 

We consider the same training dataset of one-neutron ($S_{1n}$) and two-neutron ($S_{2n}$) separation energies AME2003 \citep{AME2003} restricted to the calcium (Ca) region on the nuclear landscape with $Z\ge 14$ and $N \le 22$ ($n = 139$). The predictive performances of each model augmented with a GP model for systematic discrepancies and the BMA posterior mean predictor are evaluated on both the training dataset and a testing dataset of new measurements in AME2016 ($n = 14$) \citep{AME2016}. Similarly to \cite{Neufcourt18-2} 
we calculate model posterior probabilities independently over four non-overlapping nuclear domains according to the parity of numbers $Z$ and $N$ with uniform prior distribution over the model space. We assess the performance of the BMA using RMSE improvement and empirical coverage probability, aggregated over parities in order to mitigate the relatively small size of each subset. The GP model specification and sample sizes breakdown based on the parity of $Z$ and $N$ are given in Appendix \ref{annex-examples}, and we refer you to \cite[Statistical analysis]{Neufcourt18-2} for more details.

\setlength{\tabcolsep}{4pt}
\renewcommand{\arraystretch}{0.7}
\begin{table}[!h]
\caption{Model posterior weights for the 9 nuclear mass models under consideration with RMSE (in MeV) and $\widehat{r}^2_{BMA}$ values for training (AME2003) and testing (AME2016 $\setminus$ AME2003) datasets. The last three rows correspond respectively to averaging with prior weights, simplified BMA \citep{Neufcourt18-2}, and full BMA.
\label{table-rmse-full}}
\begin{tabular}{l|l|l|l|l|ll|ll}
 \hline \noalign{\smallskip}
\textbf{} & \multicolumn{4}{c|}{\textbf{Model posterior weights}} & \multicolumn{4}{c}{\textbf{Errors}} \\ \noalign{\smallskip} \hline \noalign{\smallskip}
\textbf{} & \multicolumn{2}{c|}{\boldmath$S_{1n}$ (\textbf{odd N})} & \multicolumn{2}{c|}{\boldmath$S_{2n}$ (\textbf{even N})} & \multicolumn{2}{c|}{\textbf{Training}} & \multicolumn{2}{c}{\textbf{Testing}} \\\noalign{\smallskip} \hline \noalign{\smallskip}
\textbf{Model} & \textbf{even Z} & \textbf{odd Z} & \textbf{even Z} & \textbf{odd Z}  & \textbf{RMSE} & \boldmath$\widehat{r}^2_{BMA}$ & \textbf{RMSE} & \boldmath$\widehat{r}^2_{BMA}$ \\
\noalign{\smallskip} \hline \noalign{\smallskip}
\textbf{SLy4} & 0.000 & 0.000 & 0.000 & 0.008 & 0.076 & - &  0.713 & 0.313 \\ 
\textbf{SkP} & 0.000 & 0.000 & 0.000 & 0.000 & 0.127 & 0.308 & 0.989 & 0.642 \\
\textbf{SkM*} & 0.000 & 0.000 & 0.000 & 0.000 & 0.142 & 0.449 & 0.924 & 0.591 \\
\textbf{SV-min} & 0.000 & 0.000 & 0.000 & 0.001 & 0.107 & 0.023 & 0.840 & 0.505\\
\textbf{UNEDF0} & 0.000 & 0.009 & 0.000 & 0.000 & 0.136 & 0.400 & 0.809 & 0.466\\
\textbf{UNEDF1} & 0.845 & 0.669 & 0.000 & 0.089 & 0.110 & 0.077 & 0.550 &  -\\
\textbf{UNEDF2} & 0.002 & 0.013 & 0.000 & 0.125 & 0.109 & 0.058 & 0.806 & 0.462 \\
 \textbf{FRDM-2012} & 0.153 & 0.308 & 0.902 & 0.310  & 0.114 & 0.149 & 0.808 & 0.465 \\
 \textbf{HFB-24} & 0.000 & 0.001 & 0.098 & 0.467 & 0.146 & 0.477 & 0.806 & 0.463 \\
 \noalign{\smallskip} \hline \noalign{\smallskip}
$\mathcal{M}_{BMA(prior)}$ & & & & & 0.110 & 0.045 & 0.641 & 0.078 \\ 
$\mathcal{M}_{BMA(simple)}$ & & & & & 0.118 & 0.110 & 0.680 & 0.131 \\ 
$\mathcal{M}_{BMA}$ & & & & & 0.105 & - & 0.591 & - \\ \noalign{\smallskip} \hline \noalign{\smallskip}
\end{tabular}
\end{table}

Table \ref{table-rmse-full} presents the resulting posterior weights of the models, as well as RMSE and MSE improvement for both averaging procedures. 
The predictions based on the full BMA ($\mathcal{M}_{BMA}$) outperform the simplified method of \cite{Neufcourt18-2} ($\mathcal{M}_{BMA(simple)}$) by $11\%$ on the training dataset and $13 \%$ on the testing one, as measured by $\widehat{r}^2_{BMA}$.
The lowest RMSE on the training dataset was attained by SLy4 and UNEDF1 respectively for AME2016 $\setminus$ AME2003. This result should not discourage practitioner from using BMA posterior mean predictors, because the BMA methodology outlined in this paper allows for existence of a "best" model for a particular data domain. However, such a model does not account for modeling uncertainty whereas BMA does, and therefore the BMA posterior mean estimator performs consistently well irrespective of the dataset. In fact it attains the second lowest RMSE on both AME2003 and AME2016 $\setminus$ AME2003. Moreover, if we consider only a subset of the whole model space, the BMA attains the lowest root RMSE. See Table \ref{table-rmse-sub} in Appendix \ref{annex-examples} for the results with restricted model space. 
Fig. 7 in Appendix \ref{annex-examples} shows the ECP of the averaged nuclear mass emulators.
While it is not clear that the BMA has an improved ECP compared to each individual models, its ECP is certainly significantly better than the worst models and comparable to the models with highest fidelity.

\subsection{Domain-corrected averaging: a pedagogical example} \label{sub-pedagogical}

In this example we study a simulated scenario 
where two computationally expensive models with $x$-dynamics of the same order act in opposite directions.
We emulate the computer models as Gaussian processes with means determined by their first order Taylor expansions as in our \emph{rapid calibration via linearization} in Section \ref{Sec:BACM}, in the reasonable situation where numerical values are also available for the model gradients. 
For simplicity we ignore here the systematic discrepancy between the computer model and the physical model.

We consider a synthetic dataset $y$ of 18 observations 
drawn independently from a normal distribution with mean  $0$ and standard deviation $10^{-3}$ at input points $x = \{\pm k, k= 1, 2, \ldots 9\}$.
We denote $f_1$ and $f_2$  the approximate response function of the two computer models obtained from first order Taylor expansions as
\begin{align} \label{Ideal-model}
f_i(x,\theta_i) &\approx \alpha_i x^2 + \theta_i, \hspace{1cm} i \in \{1,2\},
\end{align}
where $\alpha_1 = 0.5$ and $\alpha_2 = -0.5 $  and $\theta_1$ and $\theta_2$ represent unknown parameters to be calibrated. 
The two expansions (\ref{Ideal-model}) emulate a natural scenario of competition between models,
similar to the proton potential example above, where we are uncertain about the nature of the physical law and resort to model mixing in order to account for this uncertainty.
 Additionally, we study the impact of the domain correction by assigning a different training dataset $y^{(k)}$ to the models $\mathcal{M}_1$ and $\mathcal{M}_2$,
 using seven different scenarios with proportions of shared observations ($D_{shared}$) ranging from $20\%$ to $80\%$ according to the scheme in Table \ref{Ideal-scheme}. Note the break of symmetry in the domain of $y^{(k)}$ denoted by a circle, we shall refer to those accordingly as symmetric and asymmetric scenarios.
 
{\renewcommand{\arraystretch}{0.1}
\begin{table}[!htbp]
	\centering
	\caption{Scheme depicting the observations contained in the training dataset $y^{(k)}$ of the models $\mathcal{M}_k$ used in the pedagogical example of Section \ref{sub-pedagogical} according to the proportion $D_{shared}$ of shared data. The crosses mark the values contained in the domain of each model.}
	\tabcolsep=0.15cm
	\begin{tabular}{c|ccccccccccccccccccc}
		\hline	\noalign{\smallskip}
		\multicolumn{1}{l}{} &  & \multicolumn{18}{c}{\textbf{Training dataset }\boldmath$y^{(k)}$} \\
				\noalign{\smallskip}
		\hline	\noalign{\smallskip}
		\multicolumn{1}{l}{\boldmath$D_{shared}$} & \textbf{Model} & -9 & -8 & -7 & -6 & -5 & -4 & -3 & -2 & -1 & 1 & 2 & 3 & 4 & 5 & 6 & 7 & 8 & 9 \\
				\noalign{\smallskip}
		\hline	\noalign{\smallskip}
		\multirow{2}{*}{0.2} & $\mathcal{M}_{1}$ & x & x & x & x & x & x & x & x & x & x &  &  &  &  &  &  &  &  \\
		& $\mathcal{M}_{2}$ &  &  &  &  &  &  &  &  & x & x & x & x & x & x & x & x & x & x \\
				\noalign{\smallskip}
		\hline	\noalign{\smallskip}
		\multirow{2}{*}{0.3} & $\mathcal{M}_{1}$ & $\otimes$ & $\otimes$ & $\otimes$ & $\otimes$ & $\otimes$ & $\otimes$ & $\otimes$ & $\otimes$ & $\otimes$ & $\otimes$ &  &  &  &  &  &  &  &  \\
		& $\mathcal{M}_{2}$ &  &  &  &  &  &  &  & $\otimes$ & $\otimes$ & $\otimes$ & $\otimes$ & $\otimes$ & $\otimes$ & $\otimes$ & $\otimes$ & $\otimes$ & $\otimes$ &  \\
				\noalign{\smallskip}
		\hline	\noalign{\smallskip}
		\multirow{2}{*}{0.4} & $\mathcal{M}_{1}$ &  & x & x & x & x & x & x & x & x & x & x &  &  &  &  &  &  &  \\
		& $\mathcal{M}_{2}$ &  &  &  &  &  &  &  & x & x & x & x & x & x & x & x & x & x &  \\
				\noalign{\smallskip}
		\hline	\noalign{\smallskip}
		\multirow{2}{*}{0.5} & $\mathcal{M}_{1}$ &  & $\otimes$ & $\otimes$ & $\otimes$ & $\otimes$ & $\otimes$ & $\otimes$ & $\otimes$ & $\otimes$ & $\otimes$ & $\otimes$ &  &  &  &  &  &  & \\
		& $\mathcal{M}_{2}$ &  &  &  &  &  &  & $\otimes$ & $\otimes$ & $\otimes$ & $\otimes$ & $\otimes$ & $\otimes$ & $\otimes$ & $\otimes$ & $\otimes$ & $\otimes$ &  & \\
				\noalign{\smallskip}
		\hline	\noalign{\smallskip}
		\multirow{2}{*}{0.6} & $\mathcal{M}_{1}$ &  &  & x & x & x & x & x & x & x & x & x & x &  &  &  &  &  &  \\
		& $\mathcal{M}_{2}$ &  &  &  &  &  &  & x & x & x & x & x & x & x & x & x & x &  &  \\
				\noalign{\smallskip}
		\hline	\noalign{\smallskip}
		\multirow{2}{*}{0.7} & $\mathcal{M}_{1}$&  &  & $\otimes$ & $\otimes$ & $\otimes$ & $\otimes$ & $\otimes$ & $\otimes$ & $\otimes$ & $\otimes$ & $\otimes$ & $\otimes$ &  &  &  &  &  & \\
		& $\mathcal{M}_{2}$ &  &  &  &  &  & $\otimes$ & $\otimes$ & $\otimes$ & $\otimes$ & $\otimes$ & $\otimes$ & $\otimes$ & $\otimes$ & $\otimes$ & $\otimes$ &  &  & \\
				\noalign{\smallskip}
		\hline	\noalign{\smallskip}
		\multirow{2}{*}{0.8} & $\mathcal{M}_{1}$ &  &  &  & x & x & x & x & x & x & x & x & x & x &  &  &  &  &  \\
		& $\mathcal{M}_{2}$ &  &  &  &  &  & x & x & x & x & x & x & x & x & x & x &  &  &  \\
				\noalign{\smallskip}
		\hline	\noalign{\smallskip}
	\end{tabular}
	\label{Ideal-scheme}
\end{table}
}

For each value of $D_{shared}$,
we carried out the domain-corrected procedure detailed in Section \ref{sec-domain-correction} and computed the evidence integrals $p(y^{(k)}|\mathcal{M}_k)$
as well as the corrective terms $p(y^{(-k)}|y^{(k)})$.
Also note that the approximate computation of these terms is more demanding than the computation of the evidence integrals \eqref{eq-I}, because it requires integration against a posterior distribution of parameters. 

{\renewcommand{\arraystretch}{0.5}
\begin{table}[!h] 
	\centering
	\begin{tabular}{c|llllccc}
		\noalign{\smallskip}
		\hline	\noalign{\smallskip}
			\multicolumn{1}{l}{\boldmath$D_{shared}$} & \textbf{Model} & \boldmath${RMSE}$ & \boldmath$p(y^{(k)} | \mathcal{M}_k)$ & \boldmath$p(y^{(-k)}|y^{(k)})$ &
			\boldmath$Q_{0}$ &
			\boldmath$Q$ & \multicolumn{1}{l}{\boldmath$\widehat{r}^2_{BMA}$} \\
			\noalign{\smallskip}
			\hline	\noalign{\smallskip}
			\multirow{3}{*}{0.3} & $\mathcal{M}_{1}$ & 4.69 & $2.69 \cdot 10^{-21}$ & $2.13 \cdot 10^{-16}$ & \multirow{3}{*}{0.03} & \multirow{3}{*}{0.80} &\boldmath$0.495$\\
			& $\mathcal{M}_{2}$ & 4.68 & $7.78 \cdot 10^{-20}$ & $9.25 \cdot 10^{-18}$ & & & \boldmath$0.494$  \\
			& $\mathcal{M}_{BMA(Q_0)}$ & 4.53 & - & - & & & -   \\
			& $\mathcal{M}_{BMA(Q)}$ & 3.33 & - & - & & & -   \\
			\noalign{\smallskip}
			\hline	\noalign{\smallskip}
			\multirow{3}{*}{0.5} & $\mathcal{M}_{1}$ & 4.63 & $7.79 \cdot 10^{-20}$ & $5.44 \cdot 10^{-13}$ & \multirow{3}{*}{0.02} & \multirow{3}{*}{0.61} & \boldmath$0.512$ \\
			& $\mathcal{M}_{2}$ & 4.38 & $3.29 \cdot 10^{-18}$ & $2.12 \cdot 10^{-14}$ & & & \boldmath$0.456$ \\
			& $\mathcal{M}_{BMA(Q_0)}$ & 4.29 & - & - & & & -   \\
			& $\mathcal{M}_{BMA(Q)}$ & 3.23 & - & - & & & -   \\
			\noalign{\smallskip}
			\hline	\noalign{\smallskip}
			\multirow{3}{*}{0.7} & $\mathcal{M}_{1}$ & 4.36 & $3.23 \cdot 10^{-18}$ & $1.13 \cdot 10^{-8}$ & 
			\multirow{3}{*}{0.02} &
			\multirow{3}{*}{0.72} & \boldmath$0.593$ \\
			& $\mathcal{M}_{2}$ & 3.62 & $1.45 \cdot 10^{-16}$ & $3.49 \cdot 10^{-10}$ & & & \boldmath$0.410$   \\
			& $\mathcal{M}_{BMA(Q_0)}$ & 3.54 & - & - & & & -   \\
			& $\mathcal{M}_{BMA(Q)}$ & 2.78 & - & - & & & -   \\
			\noalign{\smallskip}
			\hline	\noalign{\smallskip}
	\end{tabular}
	\caption{
	Summary of the domain corrected BMA
	analysis in the asymmetric case of the pedagogical example. 
	The RMSE (in MeV) was calculated based the set of common observations ($x \leq 5$).
	$BMA(Q)$ and $BMA(Q_0)$ represent respectively domain corrected BMA and BMA with independent model domains.
	$Q$ denotes the posterior odds ratio $p(y^{(-1)}|y^{(1)})p(\mathcal{M}_1| y^{(1)})/ [p(y^{(-2)}|y^{(2)})p(\mathcal{M}_2| y^{(2)})]$ used to draw samples from the mixture distribution \eqref{eq-bma-diff-dom}
	and 
	$Q_{0}$ is the ratio $p(\mathcal{M}_1| y^{(1)})/ p(\mathcal{M}_2| y^{(2)})$. The MSE improvement $\widehat{r}^2_{BMA}$ is w.r.t. BMA with domain correction.
	\label{Ideal-mod-diff-dom-main}}
\end{table}
}

Table \ref{Ideal-mod-diff-dom-main} gives a quantitative summary of the simulation results
in the asymmetric scenario, where the impact of the domain correction is stronger. See Table \ref{Ideal-mod-diff-dom-symm} in the Appendix for the symmetric case, where the impact of the domain correction is minor due to the symmetry of training data and the response functions.
As expected from our construction, BMA leads to a spectacular decrease of the square errors by about $50\%$.
The BMA posterior mean estimator outperforms consistently the estimators from individual models, at all proportions of shared training data. As the overlap between the two model domains increases, all RMSEs consistently decrease. The same observations hold in the symmetric case.

Domain corrected BMA ($BMA(Q)$) has consistently lower RMSE than BMA with independent model domains ($BMA(Q_0)$)
across $D_{shared}$.
We observe that the values of the corrective factors increase exponentially towards 1 as $D_{shared}$ increases; indeed the extreme case $D_{shared} = 1$, where both models are defined on the same domain, corresponds to the classical BMA framework where no correction is needed. 
The odds ratios stay expectedly close to 1, 
due to the fact that the deviations from out-of-domain data are comparable across the models under consideration; still the domain-corrected odds ratio $Q$ has a consistently larger variability than $Q_{0}$, 
the difference vanishing as the proportion $D_{shared}$ of data shared between the two models increases.

\begin{figure}[H] 
	\begin{centering}
		\includegraphics[width=.7\textwidth]{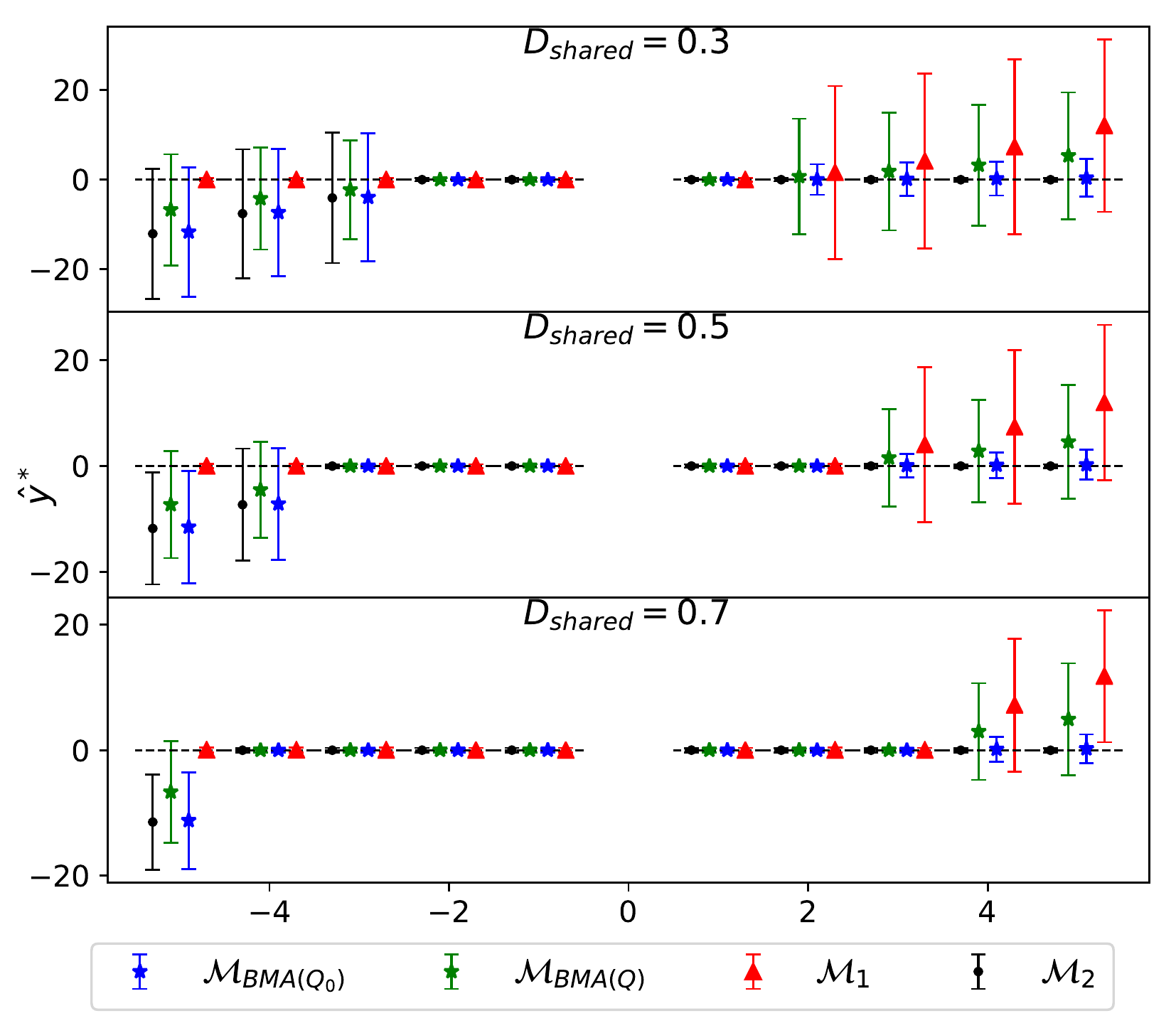}
		\caption{Posterior mean predictions (with 1-sigma error bars) for the 10 observations $y$ 
		for the two models in \eqref{Ideal-model} as well as their BMA, with domain correction ($BMA(Q)$) and with the assumption of independent model domains ($BMA(Q_0)$). 
		The dashed line segments represent the translated values of the original observations.}
		\label{fig:Idealized-mean-ratio}
	\end{centering}
\end{figure}

\section{Conclusion}

Motivated by nuclear physics research problems, we analyzed the Bayesian Model Averaging setup - the natural Bayesian framework to infer any unknown quantity of interest when several models are competing. 
We focused on the specific challenges arising from BMA of computer models such as the calibration of both computationally inexpensive and expensive models as well as the computation of the evidence integral in this context.
The nuclear physics perspective led us to study the special case of averaging models with different domains, which has not been thoroughly explored in the literature to our knowledge. 
We also gave a theoretical justification for the use of BMA posterior mean predictor in terms of PMSE reduction. While this predictor does not guarantee a universal improvement in predictive ability, on average, it performs at least as well as the best models under consideration.
Finally, we applied the methodology outlined in this paper under several scenarios that lead to considerable MSE reduction; One simple and transparent exercise of averaging of proton potentials and a pedagogical example of domain-corrected averaging with synthetic dataset. We hope that these illustrative examples provide insight into the benefits of BMA and serve as a "proof of concept." We also provided a full-scale BMA analysis of 9 state-of-the-art nuclear mass models. 

Fully documented Python code that reproduces all the examples in this paper is available at \url{https://github.com/kejzlarv/BA_of_computer_models} and can be easily modified to serve practitioners.

There are several opportunities to further explore BMA of computer models, within and beyond nuclear physics. As any other Bayesian method, it hinges on efficient sampling from posterior distributions. Direct sampling from the BMA posterior distribution could significantly improve the ease of implementation. While our theoretical basis for BMA comes from potential reduction of PMSE, a more universal argument could seek consistency of an estimator of $\Delta$ based on its BMA posteriors. Finally our domain correction to BMA corresponds to an elementary and constrained localization. Developing a more elaborated local BMA procedure could answer a wider range of practical challenges in model mixing.

\paragraph{Acknowledgement.} The authors are grateful to Witold Nazarewicz for insightful discussions on model averaging problems in nuclear physics as well as valuable comments on early versions of this work.



\newpage
\begin{appendices}

\section{BMA and prediction error lemma}

\begin{lemma*} 
For every $\lambda_1, \lambda_2\geq0$ satisfying $\lambda_1 + \lambda_2 = 1$, we have
\begin{equation}
\small
\mathbb{E}[(y^*- \widehat{y^*})^2 |y] 
= 
\mathbb{E}[(y^*- \lambda_1 \widehat{y_1^*} - \lambda_2 \widehat{y_2^*})^2 |y]
- 
[ (\lambda_1-p_1) \widehat{y_1^*} - (\lambda_2-p_2) \widehat{y_2^*} ]^2.
\end{equation}
\end{lemma*}
\begin{proof}
This follows from taking conditional expectation in the following expression, derived from standard factorization identities, and notice that the right hand side is $y$-measurable :
\begin{center}
$(y^*-\widehat{y^*})^2 - (y^*- \lambda_1 \widehat{y_1^*} - \lambda_2 \widehat{y_2^*})^2$
$= [2y^* - (\lambda_1+p_1) \widehat{y_1^*} - (\lambda_2+p_2) \widehat{y_2^*}]
[(\lambda_1-p_1) \widehat{y_1^*} + (\lambda_2-p_2) \widehat{y_2^*}]$.
\end{center}
\end{proof}

\section{Averaging models with different domains: the general case}\label{annex-mixing-domains}

Suppose given data $(x_i,y_i)_{i=1}^n$ and $K$ models $(\mathcal{M}_k)_{k=1}^K$, and assume that each model $\mathcal{M}_k$ is defined on a subset $x^{(k)}$ of $x$. 
Denote also $y^{(k)}$ the subset of $y$ corresponding to $x^{(k)}$, $y^{(-k)}$ the complementary subset as well as $y^{(\oslash)}:=\bigcap_{k}y^{(k)}$. Suppose naturally that all data locations are in the domain of at least one model so that $y=\bigcup_{k}y^{(k)}$. 
This situation arises naturally in nuclear physics, where $y_i$ is the value of a given observable (e.g. energy) corresponding to a nuclear configuration $x_i:=(Z_i,N_i)$ and where commonly used models are defined on a restricted nuclear domain where specific physical interactions prevail. \\

Note that if the datasets are disjoint, there is simply no basis to compare the models.
Given a set of models, the reader can easily get convinced that one can define a unique \emph{minimal} equivalence relationship $\star$ on the models (i.e. with a number of equivalence classes maximal) 
satisfying $\mathcal{M}\star\mathcal{M}'$ if $\mathcal{M}$ and $\mathcal{M}'$ 
share at least one data point, i.e. $\mathcal{M} \star \mathcal{M}'$ if and only if there exists $r\geq 0$ and a sequence of models 
$\mathcal{M}=:\mathcal{M}_1,\mathcal{M}_2,\ldots,\mathcal{M}_r := \mathcal{M}'$ such that $\mathcal{M}_i$ and $\mathcal{M}_{i+1}$ have a common data point for each $0\leq i < r$.
The computation of the posterior weights of the models 
can then be done within each class of equivalence, and we will therefore assume that there is only one such equivalence class. \\

In 'classical' Bayesian model averaging where all models share the same domain $y$ one can express the posterior probabilities on the models $p(\mathcal{M}_k|y)$ using Bayes formula
\begin{equation} \label{eq-master-k}
p(\mathcal{M}_k|y) \propto_k p(y|\mathcal{M}_k)\pi(\mathcal{M}_k) 
\end{equation}
and estimates the evidence integral (Bayes factor) $p(y|\mathcal{M}_k)$ as detailed in Section \ref{sub-Bayes-factor}.
In our situation, however, the model $\mathcal{M}_k$ provides an expression $p(y^{(k)}|\mathcal{M}_k)$ instead of $p(y|\mathcal{M}_k, \phi_k)$, so that the standard procedure cannot be applied without a further argument. \\

Starting from \eqref{eq-master-k}, we expand $p(y|\mathcal{M}_k)$ similarly to the two-model case as
\begin{eqnarray*}
	p(y|\mathcal{M}_k)
	& = & p(y^{(k)}, y^{(-k)}|\mathcal{M}_k) \\
	& = & p(y^{(-k)}|y^{(k)}, \mathcal{M}_k) p(y^{(k)}|\mathcal{M}_k).
\end{eqnarray*}
Instead of setting $p(y^{(-k)}|y^{(k)}, \mathcal{M}_k) = 1$ which unsatisfactorily advantages models which withhold' their predictions at difficult locations (see the example scenarios and discussion in Section \ref{sec-domain-correction}) our domain-corrected model averaging estimates
$$p(y^{(-k)}|y^{(k)}, \mathcal{M}_k) = p(y^{(-k)}|y^{(k)}).$$
This yields evidence and posterior weights given respectively by
\begin{eqnarray}
p(y|\mathcal{M}_k)
& = &
p(y^{(-k)}|y^{(k)}) p(y^{(k)}|\mathcal{M}_k) \\
p(\mathcal{M}_k|y) 
& \propto_k & 
p(y^{(-k)}|y^{(k)}) p(y^{(k)}|\mathcal{M}_k)\pi(M_k),
\end{eqnarray}
similarly to the two-model case. All that is left now is to evaluate the \emph{corrective likelihood} $p(y^{(-k)}|y^{(k)})$. \\

Letting $\mathcal{S}$ be the subset of models constraining $y^{(-k)}$ 
we can compute the \emph{corrective likelihood} $p(y^{(-k)}|y^{(k)})$ by conditioning with respect to the model:
\begin{eqnarray*}
	p(y^{(-k)}|y^{(k)}) 
	& = & \sum_{l=1}^p p(y^{(-k)}|y^{(k)}, \mathcal{M}_l)\pi(M_l) \\
	& = & \sum_{l\in\mathcal{S}} p(y^{(-k)}|y^{(k)}, \mathcal{M}_l)\pi(\mathcal{M}_l) + \sum_{l\notin\mathcal{S}}p(y^{(-k)}|y^{(k)})\pi(\mathcal{M}_l)\\
	& = & \frac{1}
	{\sum_{l\in\mathcal{S}}\pi(\mathcal{M}_l)}
	\sum_{l\in\mathcal{S}} p(y^{(-k)}|y^{(k)}, \mathcal{M}_l)\pi(\mathcal{M}_l)
\end{eqnarray*}

\noindent
$(a)$ The simple case is when $y^{(-k)}$ is \emph{non-divisible}, in the sense that for every $l$ we have $y^{(-k)}\subset y^{(l)}$ or$y^{(-k)}\cap y^{(l)} = \emptyset$: in that case, $p(y^{(-k)}|y^{(k)}, \mathcal{M}_l)$ and the sum above have explicit expressions.\\

\noindent
$(b)$ In the general case some models may be defined only on a strict subset of $y^{(-k)}$. 
In that case we have
\begin{eqnarray*}
	p(y^{(-k)}|y^{(k)}, \mathcal{M}_l) 
	& = & p(y^{(-k)\cap(l)}, y^{(-k)\cap(-l)}|y^{(k)}, \mathcal{M}_l) \\
	& = & p(y^{(-k)\cap(l)}|y^{(k)}, \mathcal{M}_l) p(y^{(-k)\cap(-l)}|y^{(k)}, y^{(-k)\cap(l)}, \mathcal{M}_l) \\
	& = & p(y^{(-k)\cap(l)}|y^{(k)\cap(l)}, \mathcal{M}_l) p(y^{(-k)\cap(-l)}|y^{(l)}, y^{(k)\cap(-l)})
\end{eqnarray*}
\noindent
The first term is explicit given a model. Hence we can compute inductively 
$p(y^{(-k_1)\cap(-k_2)\cap\ldots\cap(-k_q)}|y^s)$
for all $q$-tuples $(k_1,\ldots,k_q)$ and subset of data $y^s \subset y$ with a decreasing recursion on $q$, where the first iteration corresponds to the simple case $(a)$. \\

\noindent
For practicality purposes it is important to notice that the complexity of the underlying algorithm is at most exponential in the number of models, where each iteration requires the computation of a predictive posterior of decreasing subsets of data given decreasing subsets of the data, plus $N$ computations of corrective likelihoods as in $(a)$, where $N$ is the number of $\emph{non-divisible}$ subsets.

\section{Examples and applications charts} \label{annex-examples}

\subsection{Averaging nuclear mass emulators in the Ca region}

\subsubsection{GP model specifications}

Given a theoretical nuclear model $y^{th}$ for the one- and two-neutron separation energies
we define the discrepancy function $\delta(x)$ from 
\begin{equation}
    y(x)=y^{th}(x) + \delta(x),
\end{equation}
for $x:=(Z,N)$ ranging over the two-dimensional nuclear domain,
and model with with the GP 
\begin{equation}
    \delta(Z,N) \sim \mathcal{GP}(0, k_{\eta, \rho}\{(Z, N), (Z^{\prime}, N^{\prime})\}),
\end{equation}
with mean $0$ and quadratic exponential covariance kernel with three parameters 
\begin{equation}
     k_{\eta, \rho}\{(Z, N), (Z^{\prime}, N^{\prime})\} = \eta^2 e^{-\frac{(Z-Z^{\prime})^2}{2 \rho^2_Z} -\frac{(N-N^{\prime})^2}{2 \rho^2_N}},
\end{equation}
with Gamma prior distributions
\begin{equation}
\eta, \rho_Z, \rho_N  \sim \Gamma(a, 1)
\end{equation}
with hyperprior parameters $b=1$ and $a$ respectively set to $0.8$, $0.5$ and $1.8$.
Viz. supplemental material to \cite{Neufcourt18-2}.

\subsubsection{Supplementary figures and tables}

\begin{figure}[H]
\centering
\begin{subfigure}[a]{1\textwidth} 
\includegraphics[width=0.9\linewidth]{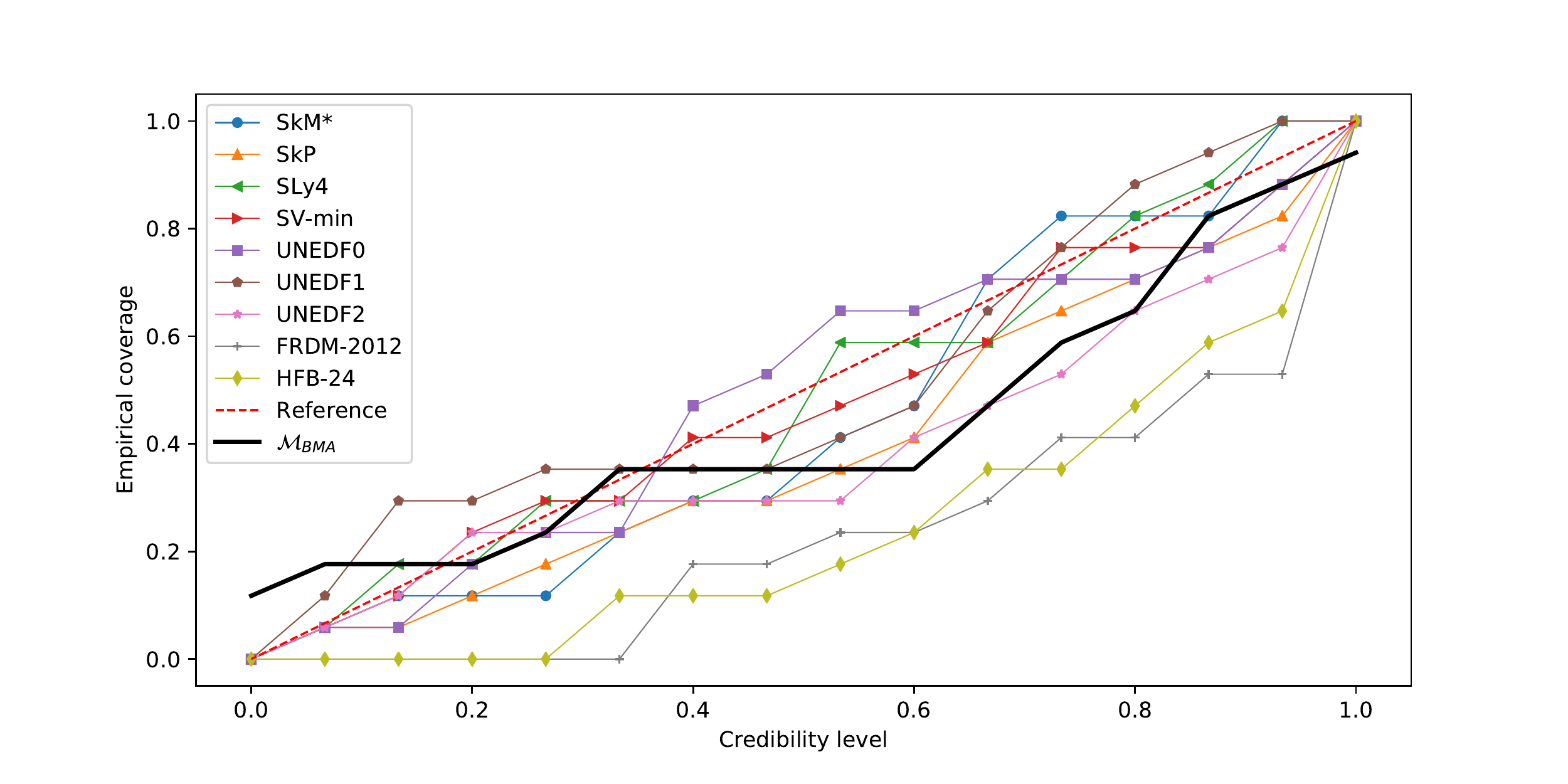}
\label{fig:ECP}
\end{subfigure}
\caption{Empirical coverage probability calculated on the independent testing dataset (AME2016 $\setminus$ AME2003).}
\end{figure}

\setlength{\tabcolsep}{4pt}
\renewcommand{\arraystretch}{0.7}
\begin{table}[!h]
\caption{Model posterior weights, RMSE (in MeV) and $\widehat{r}^2_{BMA}$ values calculated on both training (AME2003) and testing (AME2016 $\setminus$ AME2003) datasets,
for 3 nuclear mass models.
\label{table-rmse-sub}}
\begin{tabular}{l|l|l|l|l|ll|ll}
 \hline \noalign{\smallskip}
\textbf{} & \multicolumn{4}{c|}{\textbf{Model posterior weights}} & \multicolumn{4}{c}{\textbf{Errors}} \\ \noalign{\smallskip} \hline \noalign{\smallskip}
\textbf{} & \multicolumn{2}{c|}{\boldmath$S_{1n}$ (\textbf{odd N})} & \multicolumn{2}{c|}{\boldmath$S_{2n}$ (\textbf{even N})} & \multicolumn{2}{c|}{\textbf{Training}} & \multicolumn{2}{c}{\textbf{Testing}} \\\noalign{\smallskip} \hline \noalign{\smallskip}
\textbf{Model} & \textbf{even Z} & \textbf{odd Z} & \textbf{even Z} & \textbf{odd Z}  & \textbf{RMSE} & \boldmath$\widehat{r}^2_{BMA}$ & \textbf{RMSE} & \boldmath$\widehat{r}^2_{BMA}$ \\
\noalign{\smallskip} \hline \noalign{\smallskip}
\textbf{SkM*} & 0.000 & 0.001 & 0.000 & 0.000 & 0.142 & 0.375 & 0.925 & 0.413 \\
 \textbf{FRDM-2012} & 1.000 & 0.997 & 0.900 & 0.399  & 0.114 & 0.031 & 0.808 & 0.231 \\
 \textbf{HFB-24} & 0.000 & 0.002 & 0.100 & 0.601 & 0.146 & 0.405 & 0.806 & 0.227 \\
 \noalign{\smallskip} \hline \noalign{\smallskip}
$\mathcal{M}_{BMA}$ & & & & & 0.112 & - & 0.709 & - \\ \noalign{\smallskip} \hline \noalign{\smallskip}
\end{tabular}
\end{table}

\setlength{\tabcolsep}{4pt}
\renewcommand{\arraystretch}{0.7}
\begin{table}[!h]
\caption{Sample size breakdown for the training (AME2003) and testing (AME2016 $\setminus$ AME2003) datasets of nuclear separation energies in the Calcium region,
according to $Z$ and $N$ parities.
\label{table-size}}
\begin{tabular}{l|l|l|l|l}
 \hline \noalign{\smallskip}
\textbf{} & \multicolumn{4}{c}{\textbf{Sample Size}} \\ \noalign{\smallskip} \hline \noalign{\smallskip}
\textbf{} & \multicolumn{2}{c|}{\boldmath$S_{1n}$ (\textbf{odd N})} & \multicolumn{2}{c}{\boldmath$S_{2n}$ (\textbf{even N})} \\ \noalign{\smallskip} \hline \noalign{\smallskip}
\textbf{Dataset} & \textbf{even Z} & \textbf{odd Z} & \textbf{even Z} & \textbf{odd Z} \\
\noalign{\smallskip} \hline \noalign{\smallskip}
\textbf{AME2003} & 41 & 31 & 39 & 28  \\
 \textbf{AME2016 $\setminus$ AME2003} & 3 & 3 & 3 & 5   \\
 \noalign{\smallskip} \hline \noalign{\smallskip}
\end{tabular}
\end{table}

\subsection{Domain-corrected BMA}

Table \ref{Ideal-mod-diff-dom-symm} gives a quantitative summary of the simulation results
in the symmetric scenario. Figure \ref{fig:Idealized-mean-ratio-app-symm} shows the posterior mean predictions for $\mathcal{M}_1$, $\mathcal{M}_2$, domain corrected BMA $\mathcal{M}_{BMA(Q)}$, and BMA with independent model domains $\mathcal{M}_{BMA(Q_0)}$. These were obtained for the pedagogical example \ref{sub-pedagogical} using the domain correction developed in Section \ref{sec-domain-correction}. The RMSE for both $BMA(Q_0)$ and $BMA(Q)$ is almost identical here (up to a roundoff error) due to the symmetric nature of both training dataset $y^{(k)}$ and the response functions (\ref{Ideal-model}).
{\renewcommand{\arraystretch}{0.5}
\begin{table}[!h] 
	\centering
	\begin{tabular}{c|llllccc}
		\noalign{\smallskip}
		\hline	\noalign{\smallskip}
			\multicolumn{1}{l}{\boldmath$D_{shared}$} & \textbf{Model} & \boldmath${RMSE}$ & \boldmath$p(y^{(k)} | \mathcal{M}_k)$ & \boldmath$p(y^{(-k)}|y^{(k)})$ &
			\boldmath$Q_{0}$ &
			\boldmath$Q$ & \multicolumn{1}{l}{\boldmath$\widehat{r}^2_{BMA}$} \\
			\noalign{\smallskip}
			\hline	\noalign{\smallskip}
			\multirow{3}{*}{0.2} & $\mathcal{M}_{1}$ & 4.69 & $2.78 \cdot 10^{-21}$ & $1.98 \cdot 10^{-19}$ & \multirow{3}{*}{1.02} & \multirow{3}{*}{0.96} &\boldmath$0.512$\\
			& $\mathcal{M}_{2}$ & 4.58 & $2.73 \cdot 10^{-21}$ & $2.11 \cdot 10^{-19}$ & & & \boldmath$0.488$  \\
			& $\mathcal{M}_{BMA(Q_0)}$ & 3.28 & - & - & & & -   \\
			& $\mathcal{M}_{BMA(Q)}$ & 3.28 & - & - & & & -   \\
			\noalign{\smallskip}
			\hline	\noalign{\smallskip}
			\multirow{3}{*}{0.4} & $\mathcal{M}_{1}$ & 4.64 & $7.99 \cdot 10^{-20}$ & $4.33 \cdot 10^{-16}$ & \multirow{3}{*}{1.01} & \multirow{3}{*}{1.10} & \boldmath$0.511$ \\
			& $\mathcal{M}_{2}$ & 4.53 & $7.95 \cdot 10^{-20}$ & $3.96 \cdot 10^{-16}$ & & & \boldmath$0.486$ \\
			& $\mathcal{M}_{BMA(Q_0)}$ & 3.24 & - & - & & & -   \\
			& $\mathcal{M}_{BMA(Q)}$ & 3.25 & - & - & & & -   \\
			\noalign{\smallskip}
			\hline	\noalign{\smallskip}
			\multirow{3}{*}{0.6} & $\mathcal{M}_{1}$ & 4.37 & $3.32 \cdot 10^{-18}$ & $8.59 \cdot 10^{-12}$ & 
			\multirow{3}{*}{1.01} &
			\multirow{3}{*}{1.11} & \boldmath$0.504$ \\
			& $\mathcal{M}_{2}$ & 4.33 & $3.29 \cdot 10^{-18}$ & $7.84 \cdot 10^{-12}$ & & & \boldmath$0.495$   \\
			& $\mathcal{M}_{BMA(Q_0)}$ & 3.07 & - & - & & & -   \\
			& $\mathcal{M}_{BMA(Q)}$ & 3.08 & - & - & & & -   \\
			\noalign{\smallskip}
			\hline	\noalign{\smallskip}
			\multirow{3}{*}{0.8}& $\mathcal{M}_{1}$ & 3.61 & $1.45 \cdot 10^{-16}$ & $2.99 \cdot 10^{-6}$ & 
			\multirow{3}{*}{1.02} & \multirow{3}{*}{1.03} & \boldmath$0.509$ \\
			& $\mathcal{M}_{2}$ & 3.56 & $1.42 \cdot 10^{-16}$ & $2.98 \cdot 10^{-6}$ & & & 
			\boldmath$0.495$   \\
			& $\mathcal{M}_{BMA(Q_0)}$ & 2.53 & - & - & & & -   \\
			& $\mathcal{M}_{BMA(Q)}$ & 2.53 & - & - & & & -   \\
			\noalign{\smallskip}
			\hline	\noalign{\smallskip}
	\end{tabular}
	\caption{
	Summary of the domain corrected BMA
	analysis in the asymmetric case of the pedagogical example. 
	The RMSE (in MeV) was calculated based the set of common observations ($x \leq 5$).
	$BMA(Q)$ and $BMA(Q_0)$ represent respectively domain corrected BMA and BMA with independent model domains.
	$Q$ denotes the posterior odds ratio $p(y^{(-1)}|y^{(1)})p(\mathcal{M}_1| y^{(1)})/ [p(y^{(-2)}|y^{(2)})p(\mathcal{M}_2| y^{(2)})]$ used to draw samples from the mixture distribution \eqref{eq-bma-diff-dom}
	and $Q_{0}$ is the ratio $p(\mathcal{M}_1| y^{(1)})/ p(\mathcal{M}_2| y^{(2)})$. The MSE improvement $\widehat{r}^2_{BMA}$ is w.r.t. BMA with domain correction.
	\label{Ideal-mod-diff-dom-symm}
	}.
	
\end{table}
}

\begin{figure}[H] 
	\begin{centering}
		\includegraphics[width=0.75\textwidth]{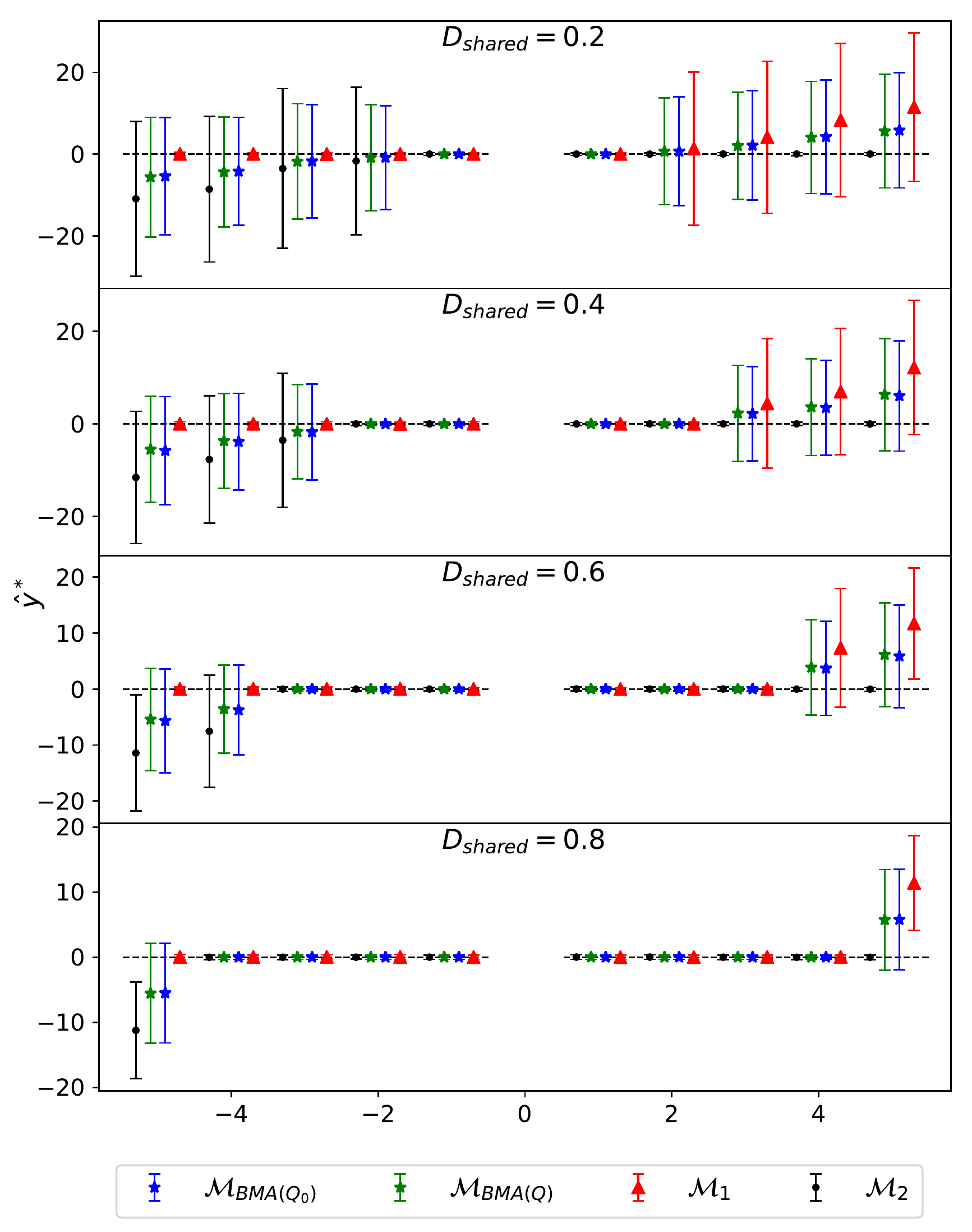}
		\caption{Posterior mean predictions (with 1-sigma error bars) for the 10 observations $y$ for the two models in \eqref{Ideal-model} as well as their BMA, with domain correction ($BMA(Q)$) and with the assumption of independent model domains ($BMA(Q_0)$). The dashed line segments represent the translated values of the original observations.}
		\label{fig:Idealized-mean-ratio-app-symm}
	\end{centering}
\end{figure}

\end{appendices}
\newpage
\bibliography{biblio}
\end{document}